\documentclass[11pt]{article}
\usepackage[utf8]{inputenc}

\usepackage{amsmath,amsfonts,amssymb,amsthm}
\usepackage[margin=1in]{geometry}
\usepackage{graphicx}
\usepackage{float}
\usepackage[colorlinks]{hyperref}
\usepackage[capitalize]{cleveref}

\usepackage{mathrsfs}
\usepackage{mathtools}
\mathtoolsset{centercolon}
\usepackage{bm}
\usepackage{multirow}
\usepackage{bbm}
\usepackage{subcaption}
\usepackage[qm]{qcircuit}
\usepackage{adjustbox}

\newcommand{\C}{{\mathbb{C}}}
\newcommand{\E}{{\mathbb{E}}}

\newcommand{\abs}[1]{\left\lvert#1\right\rvert}

\def\E{{\mathbb E}}
\def\V{{\mathrm{Var}}}

\def \Im{\mathop{\rm Im}}
\def \Re{\mathop{\rm Re}}
\def\C{{\mathbb C}}

\def\le{\leqslant}
\def\ge{\geqslant}

\newcommand{\Or}{\mathcal{O}}

\numberwithin{equation}{section}
\DeclareMathOperator{\polylog}{polylog}

\renewcommand{\d}{\mathrm{d}}

\renewcommand{\Re}{\mathop{\mathrm{Re}}}
\renewcommand{\Im}{\mathop{\mathrm{Im}}}

\newcommand{\mc}[1]{\mathcal{#1}}

\newcommand{\wt}[1]{\widetilde{#1}}
\newcommand{\wh}[1]{\widehat{#1}}

\newcommand{\norm}[1]{\left\lVert#1\right\rVert}

\newcommand{\ud}{\,\mathrm{d}}

\newtheorem{problem}{Problem}
\newtheorem{theorem}{Theorem}[section]
\newtheorem{lemma}[theorem]{Lemma}

\theoremstyle{definition}

\numberwithin{equation}{section}
\allowdisplaybreaks

\numberwithin{equation}{section}

\newcommand{\eq}[1]{(\ref{eq:#1})}
\renewcommand{\sec}[1]{\hyperref[sec:#1]{Section~\ref*{sec:#1}}}
\newcommand{\app}[1]{\hyperref[app:#1]{Appendix~\ref*{app:#1}}}
\newcommand{\thm}[1]{\hyperref[thm:#1]{Theorem~\ref*{thm:#1}}}
\newcommand{\lem}[1]{\hyperref[lem:#1]{Lemma~\ref*{lem:#1}}}
\newcommand{\cor}[1]{\hyperref[cor:#1]{Corollary~\ref*{cor:#1}}}
\newcommand{\prb}[1]{\hyperref[prb:#1]{Problem~\ref*{prb:#1}}}
\newcommand{\fgr}[1]{\hyperref[fgr:#1]{Figure~\ref*{fgr:#1}}}
\newcommand{\tab}[1]{\hyperref[tab:#1]{Table~\ref*{tab:#1}}}

\newcommand{\beq}{\begin{equation}}
\newcommand{\eeq}{\end{equation}}
\newcommand{\beqa}{\begin{eqnarray}}
\newcommand{\eeqa}{\end{eqnarray}}
\newcommand{\bra}[1]{\ensuremath{{\langle{#1}|}}}
\newcommand{\ket}[1]{\ensuremath{{|{#1}\rangle}}}
\newcommand{\braket}[1]{\ensuremath{{\langle{#1}\rangle}}}

\usepackage{soul}

\makeatletter
\let\@@magyar@captionfix\relax
\makeatother

\title{Dense outputs from quantum simulations}
\author{Jin-Peng Liu$^{1,2,3,\dag}$, \quad Lin Lin$^{1,4,5,*}$\\
\footnotesize $^{1}$ Department of Mathematics, University of California, Berkeley\\
\footnotesize $^{2}$ Simons Institute for the Theory of Computing, University of California, Berkeley\\
\footnotesize $^{3}$ Center for Theoretical Physics, Massachusetts Institute of Technology, Cambridge\\
\footnotesize $^{4}$ Applied Mathematics and Computational Research Division, Lawrence Berkeley National Laboratory, Berkeley\\
\footnotesize $^{5}$ Challenge Institute for Quantum Computation, University of California, Berkeley, Berkeley\\
\footnotesize Email: 
$^{\dag}$ jliu1219@terpmail.umd.edu \footnotesize $^{*}$ linlin@math.berkeley.edu
}
\date{}

\begin{document}

\maketitle

\begin{abstract}
The quantum dense output problem is the process of evaluating time-accumulated observables from time-dependent quantum dynamics using quantum computers. This problem arises frequently in applications such as quantum control and spectroscopic computation. We present a range of algorithms designed to operate on both early and fully fault-tolerant quantum platforms. These methodologies draw upon techniques like amplitude estimation, Hamiltonian simulation, quantum linear Ordinary Differential Equation (ODE) solvers, and quantum Carleman linearization. We provide a comprehensive complexity analysis with respect to the evolution time $T$ and error tolerance $\epsilon$.  Our results demonstrate that the linearization approach can nearly achieve optimal complexity $\mathcal{O}(T/\epsilon)$ for a certain type of low-rank dense outputs. Moreover, we provide a linearization of the dense output problem that yields an exact and finite-dimensional closure which encompasses the original states. This formulation is related to the Koopman Invariant Subspace theory and may be of independent interest in nonlinear control and scientific machine learning. 
\end{abstract}

\section{Introduction}\label{sec:introduction}

Simulating quantum physics is one of the primary applications of quantum computers~\cite{Fey82}. The first explicit quantum algorithm for quantum simulation was proposed by Lloyd~\cite{Llo96} using product formulas, and numerous quantum simulation algorithms have been developed~\cite{Wie96,Zal98,KWB17,BAC07,PQS11,BCC13,BCC15,BCK15,LC16,LC17,LW18,CMN17,COS19,BCS20,CST21,SBW21,SHC21,AFL20,AFL22,ZZS21,CCH22,CLLL22}, with various applications ranging from quantum chemistry~\cite{KJL08,Lan10,MEA18,CRO18,BBMC20,SBW21} to quantum field theory~\cite{JLP11,pre19} and condensed matter physics~\cite{BWC17}.
To analyze the cost of these quantum simulation algorithms, it is often assumed that we are interested in the final quantum state at some time $T$. After obtaining such a state stored in a quantum register, we can then output its information to a classical computer by measuring certain observables. 

However, many applications require not only the information at the final simulation time $T$, but the information of the quantum state on a continuous time interval, or its discretized form with dense samples on the time interval (the detailed definition for the evaluation of a time-accumulated observable refers to \prb{dense} below).
Unless the observable of interest commutes with the Hamiltonian, according to the principles of quantum mechanics, the state collapses after each measurement. As a result, when we acquire the state at time $t_n$ and conduct a measurement, it becomes necessary to restart the simulation from time $0$ in order to perform additional measurements at $t_n$ or to obtain the state at $t_{n+1}$.
By following this straightforward algorithm, we treat the simulation up to each discretized time step $t_n$ as an individual quantum simulation problem.
It is natural to ask whether there are more efficient algorithms than the strategy above for evaluating time-accumulated observables. 
Motivated by the literature of classical simulation of differential equations with dense outputs (see e.g.,~\cite{HairerOstermann1990,LandryCaboussatHairer2009})
we refer to this setting as \textit{dense outputs from quantum simulations}.\footnote{
Classical Hamiltonian simulation algorithms typically operate with the wave-function (and consequently, observables) defined on a discrete time grid. Consequently, the concept of \textit{dense outputs} takes on a slightly different meaning, wherein observables must be resolved on an even finer time grid compared to the given time grid.}
To our knowledge, such a setting has not been analyzed before in the quantum algorithms literature.

We formally define the quantum dense output problem as below.

\begin{problem}[Quantum dense output]\label{prb:dense}
A time-accumulated observable associated with a time-dependent quantum dynamics is given by
\begin{equation}
\begin{aligned}
    \frac{\d }{\d t}\ket{\psi(t)} &= -iH(t) \ket{\psi(t)}, \quad \ket{\psi(0)}=\ket{\psi_{\mathrm{in}}},\\
    J &= \int_0^T \bra{\psi(t)}O(t)\ket{\psi(t)}  ~\d t.
\label{eq:dense}
\end{aligned}
\end{equation}
Here $H(t)$ and $O(t)$ represent continuous Hermitian matrices with respect to $t$. We have access to unitaries that block encode $H(t)$ and $O(t)$ for all $t$, with $\|O(t)\|\le 1$. Additionally, we are provided with a state preparation oracle that prepares the initial state $\ket{\psi_{\mathrm{in}}}$. Our objective is to estimate the value of $J$ with a desired precision of $\epsilon$, within a given time duration $T > 0$.
\end{problem}

\prb{dense} arises in diverse areas such as quantum control and spectroscopic computation. For instance, Li and Wang studied efficient quantum algorithms for the quantum control problem of the Mayer type~\cite{LW23}. As a more general case, \eq{dense} can be viewed as a quantum control problem of the Bolza type~\cite{Dal21,BCR10,RWP09,WG07}. The dense output problem can also arise when $O(t) = \ket{\phi(t)}\bra{\phi(t)}$, and $\int_0^T \bra{\psi(t)}O(t)\ket{\psi(t)} ~\mathrm{d}t$ gives the time accumulated fidelity between the driven state $\ket{\psi(t)}$ and the desired state or trajectory $\ket{\phi(t)}$.
In spectroscopic computation, the spectra estimation from the molecular dynamics can also be formulated as \eq{dense}. The observable $\int_0^T e^{i\omega t}\bra{\psi(t)}O(t)\ket{\psi(t)} ~\d t$ can be used to compute the linear absorption and fluorescence spectra of molecular aggregates~\cite{ren2018time}. 
For further discussions on the applications of \prb{dense}, we refer readers to \sec{application}. 

We expect the optimal complexity of solving \prb{dense} is $\Or(T/\epsilon)$ for general Hamiltonian systems, since there is a lower bound $\Omega(T)$ by no-fast-forwarding theorem~\cite{BAC07} as well as a lower bound $\Omega(1/\epsilon)$ by the Heisenberg limit~\cite{AA17}. To achieve this goal, we develop several quantum algorithms for \prb{dense} as summarized in \tab{summary}.

\begin{table}[H]
\renewcommand{\arraystretch}{1.5}
    \centering
    \begin{adjustbox}{width=\textwidth}
    \begin{tabular}{c|c|c|c|c|c}
      \hline\hline
      \textbf{Theorem} & \textbf{Algorithm} & \textbf{Measurement} & \textbf{Queries to $H(t)$} & \textbf{Queries to $\ket{\psi_{\mathrm{in}}}$} & \textbf{Notes} \\
      \hline
      \thm{hs-ht} & Hamiltonian simulation & Hadamard test & $\Or(T^3/\epsilon^2)$ & $\Or(T^2/\epsilon^2)$ & Early fault-tolerant \\
      \hline
      \multirow{2}*{\thm{hs-ae}} & \multirow{2}*{Hamiltonian simulation} & Biased amplitude estimation & $\Or(T^3/\epsilon)$ & $\Or(T^2/\epsilon)$ & Fault-tolerant \\
      \cline{3-6}
      &  & Unbiased amplitude estimation & $\Or(T^{2.5}/\epsilon)$ & $\Or(T^{1.5}/\epsilon)$ & Fault-tolerant \\
      \hline
      \thm{lode-ae} & Quantum linear ODE solver & Amplitude estimation & $\Or(T^2/\epsilon)$ & $\Or(T^2/\epsilon)$ & Non-unitary relaxation \\
      \hline
      \thm{node-ae} & Quantum Carleman linearization & Padding, amplitude estimation & $\Or(\Gamma T/\epsilon)$ & $\Or(\Gamma T/\epsilon)$ & Low-rank, linearization \\
      \hline\hline
    \end{tabular}
    \end{adjustbox}
    \caption{Summary of quantum algorithms for the dense output problem. Here $T$ is the evolution time, $\epsilon$ is the error tolerance, and $\Gamma$ is a parameter that depends on the output $J$ as defined in \eq{Gamma}.  }
    \label{tab:summary}
\end{table}

(i) We first consider an \emph{early fault-tolerant quantum algorithm}. We perform separate Hamiltonian simulations and employ Hadamard test, with complexity $\Or(T^3/\epsilon^2)$. 

(ii) We then propose \emph{fault-tolerant quantum algorithms} with improved $\epsilon$ dependence. We perform separate Hamiltonian simulations with biased and unbiased amplitude estimation, with improved complexity $\Or(T^3/\epsilon)$ and $\Or(T^{2.5}/\epsilon)$, respectively. 

(iii) In spite of quantum simulation algorithms, we alternatively consider the \emph{non-unitary relaxation}: we apply the quantum linear Ordinary Differential Equation (ODE) solver~\cite{Ber14,BCOW17,CL19} to produce the Feynman-Kitaev history state, and then perform the global amplitude estimation for the dense output. The global measurement is able to remove the bias accumulation in measurement, and hence result in the overall complexity $\Or(T^2/\epsilon)$.

(iv) Finally, we consider the \emph{non-unitary embedding/linearization} of the whole system: we develop a quantum linearization algorithm for the hybrid dynamics based on the Koopman Invariant Subspace (KIS) theory~\cite{BBPK16}. For \prb{dense} with a low-rank observable $\Or(t)$ (known as the few-body observable~\cite{HK19,HKP20,HCP22}; see \prb{qdcd} for detailed discussions), we employ an \emph{exact finite-dimensional linear representation (closure)} of the nonlinear hybrid quantum-classical dynamics. For the resulting linearized dynamics, we apply the quantum linear ODE solver and perform the amplitude estimation with padding to achieve the overall complexity $\Or(T/\epsilon)$. This result is nearly tight for both $T$ and $\epsilon$, matching the no-fast-forwarding lower bound $\Omega(T)$~\cite{BAC07} and the Heisenberg limit lower bound $\Omega(1/\epsilon)$~\cite{AA17}.

From the viewpoint of the Koopman von Neumann operator theory, the linearization of the dense output problem offers a concrete example in quantum mechanics such that: (i) it has an \emph{exact} finite-dimensional closure of the nonlinear dynamics without truncation; (ii) the closure system explicitly includes the original state variables. In the Koopman Invariant Subspace (KIS) theory, very few examples are known with both finite-dimensional exact and explicit closure, while previous examples are only found in classical mechanics, such as the polynomial attracting slow manifold~\cite{BBPK16}. We believe the quantum linearization methods for the dense output problem can be of independent interest in Koopman Operator Optimal Control (KOOC), Dynamic Mode Decomposition (DMD), and data-driven discovery and identification in scientific machine learning~\cite{LM98,Mez05,KBBP16,BK22}.

It is worth noting that the core findings of this paper on dense outputs are largely unaffected by the nature of the Hamiltonian simulation, whether it is time-dependent or not. For the sake of notational simplicity, we can focus on time-independent simulations. In cases where explicit time-dependence is present, we can assume access to the time-dependent Hamiltonian matrix  $H(t)$ through an oracle defined as
\begin{equation}
(\bra{0^{m_H}} \otimes I) \;\text{HAM-T}\; (\ket{0^{m_H}} \otimes I) = \frac{1}{\alpha_H}\sum_{l} \ket{l}\bra{l} \otimes H(t_l)
\end{equation}
over a sufficiently dense time grid $\{t_l\} \subset [0, T]$, using $m_H$ ancillary qubits. In this case, we should interpret $\norm{H} := \sup_{t \in [0, T]} \norm{H(t)}$.  For detailed discussions on such an oracle, we direct readers to~\cite{LW18, FLT22}. We also assume the ratio between the block encoding factor $\alpha_H$ and the operator norm $\norm{H}$ satisfies $\alpha_H/\norm{H} = \mathcal{O}(1)$. The analysis of differential equation solvers often involves many polylogarithmic factors. To simplify the presentation, we may slightly abuse the big-$\Or$ notation to suppress some of these polylogarithmic factors.

The rest of the paper is organized as follows. \sec{early} introduces an early fault-tolerant quantum simulation algorithm with Hadamard test. \sec{simulation} describes fault-tolerant quantum simulation algorithms with biased or unbiased amplitude estimation. \sec{ode} proposes a quantum linear ODE solver for producing the history-state solution with global amplitude estimation. \sec{linearization} develops a quantum linearization algorithm for the low-rank dense outputs and perform amplitude estimation with padding. \sec{application} discusses several applications of our algorithms, including quantum control and spectroscopic computation. Finally, we conclude and discuss open questions in \sec{discussion}.

\section{Early fault-tolerant quantum simulation algorithm}\label{sec:early}

We start by outlining an early fault-tolerant algorithm. To be specific, we expect such early fault-tolerant quantum algorithms feature a limited number of logical qubits, controlled operations, and ancilla qubits, as well as a short circuit depth. Therefore, we exploit the Hadamard test circuit. The quantum circuit is simple and uses only one ancilla qubit as required.

\begin{theorem}[Hamiltonian simulation with Hadamard test]\label{thm:hs-ht}
We consider an instance of the quantum dense out problem in \prb{dense}.
There exist a quantum algorithm producing an observable approximating the cost functional $J(u)$ with error $\epsilon\in(0,1)$, succeeding with probability $1-\delta$, with 
\begin{equation}
\Or\Bigl( \frac{\|H\|T^3\log(1/\delta)}{\epsilon^2} \Bigr).
\end{equation}
queries to the matrix oracle for $H(t)$, and 
\begin{equation}
\Or\Bigl(\frac{T^2 \log(1/\delta)}{\epsilon^2}\Bigr)
\end{equation}
queries to the state preparation oracle for $\ket{\psi_{\mathrm{in}}}$.
\end{theorem}

\begin{proof}
In \prb{dense}, we need to estimate the integral 
\begin{equation}
A=\int_0^T \braket{O}_t\ud t \coloneqq \int_0^T \braket{\psi(t)|O(t)|\psi(t)} \ud t
\end{equation}
to precision $\epsilon$ with success probability at least $1-\delta$. 

We divide the time interval $[0,T]$ with a composite Clenshaw--Curtis quadrature rule, with nodes $\{t_{1},\ldots,t_{N_t}\}$ and weights $\{\omega_{1},\ldots,\omega_{N_t}\}$. We can approximate $A$ as
\begin{equation}
\wt{A}=\sum_{k=1}^{N_t} \omega_{k} \braket{O}_{t_{k}}
\label{eq:quadrature}
\end{equation}
with $\abs{A-\wt{A}}\le \epsilon/2$, where $N_t = \Or(T\log(1/\epsilon))$. More details refer to \app{quadrature}. 

For the $k$-th iteration, we first propagate a few copies of $\ket{\psi(0)}$ to $\ket{\psi(t_k)}$ by standard Hamiltonian simulation techniques, with known complexity $\Or(\|H\|t_k) =\Or(\|H\|k\Delta t)$~\cite{BAC07,BCK15,CST21}; we then take measurements from copies of $\ket{\psi(t_k)}$ to produce the observable
\begin{equation}
f_k = \braket{O}_{t_{k}} = \bra{\psi(t_k)} O(t_k) \ket{\psi(t_k)}.
\end{equation}
The time complexity of producing such observables equals to the \emph{product} of the simulation cost and the measurement cost. This is due to the fact that the quantum state collapses after measurement, and hence one must restart propagating $\ket{\psi(0)}$ to $\ket{\psi(t_k)}$ for each $k$-th iteration.

Now for each given $t_k$, we can evaluate $f_k=\braket{O}_{t_k}$ using the Hadamard test.

\paragraph{Hadamard Test:}

Assuming each $O(t_k)$ can be accessed via a block encoding matrix, we can estimate $\Re{\bra{\psi(t_k)} O(t_k) \ket{\psi(t_k)}}$ using the standard Hadamard test circuit. We introduce a random variable $X_k$ and set it to be $1$ when the measurement outcome of the ancilla qubit is $0$, and set it to be $-1$ when the measurement outcome is $1$.
Similarly, we can estimate $\Im{\bra{\psi(t_k)} O(t_k) \ket{\psi(t_k)}}$, and introduce a random variable $Y_k$ that depends in the same way on the measurement outcome.

We then compute
\begin{equation}
f_k=\mathbb{E} X_k + \mathrm{i} \mathbb{E} Y_k,
\end{equation}
In practice, an \emph{unbiased estimator} to $f_k$ is
\begin{equation}
\wt{f}_k=\frac{1}{N_s}\sum_{l=1}^{N_s} (X_k^{(l)}+ \mathrm{i} Y_k^{(l)}),
\end{equation}
where $N_s$ is the number of samples, and $X_k^{(l)},Y_k^{(l)}$ are independent samples. Since $|\omega_k X_k^{(l)}|,|\omega_k Y_k^{(l)}|\le \omega_k$, by applying Hoeffding's inequality to the real and imaginary part of the observables respectively, we have  
\begin{equation}
\begin{split}
\mathbb{P}\left(\abs{\wt{A}-\sum_{k=1}^{N_t}\omega_k \wt{f}_k}\ge \frac{\epsilon}{2}\right)&
\le \mathbb{P}\left(\abs{\Re\wt{A}-\sum_{k=1}^{N_t}\omega_k \Re\wt{f}_k}\ge \frac{\epsilon}{2\sqrt{2}}\right)+\mathbb{P}\left(\abs{\Im\wt{A}-\sum_{k=1}^{N_t}\omega_k \Im\wt{f}_k}\ge \frac{\epsilon}{2\sqrt{2}}\right)\\
&\le 4 \exp\left(-\frac{N_s^2 \epsilon^2}{16 N_s\sum_{k=1}^{N_t} \omega_k^2}\right) = 4 \exp\left(-\frac{N_s \epsilon^2}{16 \sum_{k=1}^{N_t} \omega_k^2}\right)\\
\end{split}
\label{eq:hoeffding_had}
\end{equation}
So we need to estimate the 2-norm of the weight $\sum_{k=1}^{N_t} \omega_k^2$. 

According to \app{quadrature}, we have
\begin{equation}
\sum_{k=1}^{N_t} \omega_k^2=\Or(T/\log(1/\epsilon))=\Or(T), \quad N_t =  \Or(T\log(1/\epsilon)).
\end{equation}
Plug this into \eq{hoeffding_had}, we can choose
\begin{equation}
N_s=\Or\left(\frac{T \log(1/\delta)}{\epsilon^2}\right)
\end{equation}
so that
\begin{equation}
\mathbb{P}\left(\abs{\wt{A}-\sum_{k=1}^{N_t}\omega_k \wt{f}_k}\ge \frac{\epsilon}{2}\right)<\delta.
\end{equation}
Taking the quadrature error $\frac{\epsilon}{2}$ into account (more details refer to \app{quadrature}), we can estimate $A$ within precision $\epsilon$ with probability at least $1-\delta$. 

Since the cost for propagating $\ket{\psi(t_k)}$ is proportional to $\|H\|t_k =  \|H\|k\Delta t$, the algorithm for evaluating $A$ takes 
\begin{equation}
\Or\Bigl(\sum_{k=1}^{N_t} \|H\|k\Delta t \cdot N_s \Bigr) = \Or\Bigl( \frac{\|H\|T^3\log(1/\delta)}{\epsilon^2} \Bigr).
\end{equation}
queries to the matrix oracle for $H(t)$. 

For the state preparation, we need to prepare a number of quantum states $\Bigl\{ \ket{\psi(t_k)} \Bigr\}_{k=1}^{N_t}$ and each $\ket{\psi(t_k)}$ requires $N_s$ copies. 
Overall, the algorithm takes 
\begin{equation}
N_t\cdot N_s = \Or\Bigl(\frac{T^2 \log(1/\delta)}{\epsilon^2}\Bigr)
\end{equation}
queries to the state preparation oracle for $\ket{\psi_{\mathrm{in}}}$.
\end{proof}

\section{Fault-tolerant quantum simulation algorithm}\label{sec:simulation}

In the fully fault-tolerant quantum computation scenario, we are able to employ amplitude amplification and estimation with improved accuracy~\cite{BHM02}. We state the standard (biased) amplitude estimation as follows. 

\begin{lemma}[Theorem 12 of~\cite{BHM02}]\label{lem:nae}
Given a state $\ket{\psi}$ and reflection operators $R_{\psi} = 2 \ket{\psi}\bra{\psi} - I$ and $R = 2P-1$, and any $0<\eta<1$, there exists a quantum algorithm that outputs $\wt a$, an approximation to $a = \bra{\psi} P \ket{\psi}$, so that
\begin{equation}
| \wt a - a | \le 2\pi\frac{\sqrt{a(1-a)}}{r} + \frac{\pi^2}{r^2}.
\end{equation}
with probability at least $1-\eta$ and $\Or(r\log(1/\eta))$ uses of $R_{\psi}$ and $R$.  Moreover, if $a \le 1/(4r^2)$, then $\wt a = 0$ with probability at least $1-\eta$. 
\end{lemma}

We also consider recent advances in the field of unbiased amplitude estimation~\cite{RF22,vCGN22,CH22}. Here we employ the unbiased amplitude estimator proposed in~\cite{CH22}.  

\begin{lemma}[Theorem 2.4 of~\cite{CH22}]\label{lem:nuae}
Given a state $\ket{\psi}$ and a projection operator $\Pi$ with $p = \|\Pi \ket{\psi}\|^2$, $t\ge1$ and $\epsilon \in (0,1)$, there exists a quantum algorithm that outputs $\wt p \in[-2\pi,2\pi]$, so that 
\begin{equation}
| \E[\wt p] - p | \le \eta, \qquad \text{and} \qquad \V(\wt p) \le \frac{91p}{r^2} + \eta.
\end{equation}
The algorithm needs $O(r\log\log(r)\log(r/\eta))$ uses of the reflection operators $I - 2 \ket{\psi}\bra{\psi}$ and $I-2\Pi$.
\end{lemma}

We state the complexity results of repeating quantum simulations with biased or unbiased amplitude estimation as below.

\begin{theorem}[Hamiltonian simulation with amplitude estimation]\label{thm:hs-ae}
We consider an instance of the quantum dense output problem in \prb{dense}.
There exist quantum algorithms producing an observable approximating the cost functional $J(u)$ with error $\epsilon\in(0,1)$, succeeding with probability $1-\delta$, with the following cost:\\
(i) Using the Biased Amplitude Estimation, the algorithm requires
\begin{equation}
\Or\Bigl( \frac{\|H\|T^3\log(1/\delta)}{\epsilon} \Bigr)
\end{equation}
queries to the matrix oracle for $H(t)$, and 
\begin{equation}
\Or\Bigl(\frac{T^2 \log(1/\delta)}{\epsilon}\Bigr)
\end{equation}
queries to the state preparation oracle for $\ket{\psi_{\mathrm{in}}}$;\\
(ii) Using the Unbiased Amplitude Estimation, the algorithm requires 
\begin{equation}
\Or\Bigl( \frac{\|H\|T^{2.5}\log(1/\delta)}{\epsilon} \Bigr).
\end{equation}
queries to the matrix oracle for $H(t)$, and 
\begin{equation}
\Or\Bigl(\frac{T^{1.5} \log(1/\delta)}{\epsilon}\Bigr)
\end{equation}
queries to the state preparation oracle for $\ket{\psi_{\mathrm{in}}}$.
\end{theorem}

\begin{proof}
Similar as \thm{hs-ht}, we aim to estimate the integral
\begin{equation}
A=\int_0^T \braket{O}_t\ud t = \int_0^T \braket{\psi(t)|O(t)|\psi(t)} \ud t.
\end{equation}
Given quadrature nodes $\{t_{1},\ldots,t_{N_t}\}$ and weights $\{\omega_{1},\ldots,\omega_{N_t}\}$ as introduced in \app{quadrature}, we consider
\begin{equation}
\wt{A} = \sum_{k=1}^{N_t} \omega_{k} \braket{O}_{t_{k}}
\label{eq:quadrature-sum}
\end{equation}
as an approximation to $A$ with $\abs{A-\wt{A}}\le \epsilon/2$.

For the $k$-th time step, we first propagate $\ket{\psi(0)}$ to obtain $U_k\ket{\psi(0)} = \ket{\psi(t_k)}$ by simulating $U_k = e^{-iHt_k}$, with cost $\wt O(\|H\|k\Delta t)$. We then perform the amplitude amplification and estimation for the correlation function
\begin{equation}
\braket{O}_{t_{k}} = \bra{\psi(t_k)} O(t_k) \ket{\psi(t_k)} = \bra{\psi(t_0)} U_k^{\dagger} O(t_k) U_k \ket{\psi(t_0)}.
\end{equation}

In all, we need to repeat quantum simulations $N_t$ times, where the simulation cost at the $k$-th stage is $\Or(\|H\|k\Delta t)$ for $k\in[N_t]$. The total time complexity of producing such observables equals to the \emph{product} of the simulation cost and the measurement cost. 

Now for each given $t_k$, we can evaluate $f_k=\braket{O}_{t_k}$ using the amplitude estimation. We note that the amplitude estimation with or (almost) without bias can result in different complexities. 

\paragraph{(i) Biased Amplitude Estimation:}

We estimate $f_k$ using the standard amplitude estimation. It suffices to take $r = \Or(1/\epsilon')$  in \lem{nae} to obtain $\wt a$ that is $\epsilon'$-close to $a$, with probability at least $1-\delta'$ and $\Or(r\log(1/\delta'))$ uses of $R_{\psi}$ and $R$.
To ensure that $\sum_{k=1}^{N_t}f_k$ in \eq{quadrature-sum} can reach the target precision $\epsilon$, it suffices to estimate each $f_k$ within $\epsilon'=\epsilon/T$, with overall probability $1-\delta$, $\delta' = \delta/N_t$, $N_t = T\log(1/\epsilon)$, with 
\begin{equation}
C_d = \Or\left(\frac{T\log(1/\delta)}{\epsilon}\right)
\end{equation} 
queries to the coherent implementation of $\ket{\psi_{t_k}}$, due to the use of the reflection $I - 2 \ket{\psi_{t_k}}\bra{\psi_{t_k}}$. Combining with the cost of processing $U_k\ket{\psi(0)} = \ket{\psi(t_k)}$ as stated above,
the algorithm therefore requires
\begin{equation}
\Or\Bigl(\sum_{k=1}^{N_t} \|H\|k\Delta t \cdot C_d \cdot N_s \Bigr) = \Or\Bigl( \frac{\|H\|T^3\log(1/\delta)}{\epsilon} \Bigr)
\end{equation}
queries to the matrix oracle for $H(t)$.

Regarding the state preparation, we need to prepare a number of quantum states $\Bigl\{ \ket{\psi(t_k)} \Bigr\}_{k=1}^{N_t}$ and each $\ket{\psi(t_k)}$ requires $C_d$ copies in the circuit of biased amplitude estimation. 
Overall, the algorithm takes 
\begin{equation}
N_t\cdot C_d = \Or\Bigl(\frac{T^2 \log(1/\delta)}{\epsilon}\Bigr)
\end{equation}
queries to the state preparation oracle for $\ket{\psi_{\mathrm{in}}}$.

\paragraph{(ii) Unbiased Amplitude Estimation:}

Using the unbiased amplitude estimation in \lem{nuae}, we can afford to estimate each $f_k$ to precision $\epsilon/\sqrt{T}$ using a circuit of depth 
\begin{equation}
C_d = \Or(\frac{\sqrt{T}}{\epsilon}).
\end{equation}
To see this, let
\begin{equation}
\wt{f}_k=\frac{1}{N_s}\sum_{l=1}^{N_s} \wt{f}_k^{(l)}.
\end{equation}
In \cref{lem:nuae}, we should choose $\eta=\Theta(\epsilon'^2),r = \Theta(1/\epsilon')$, so that
\begin{equation}
| \mathbb{E}\wt f_k - f_k | \le \eta, \qquad \V(\wt f_k) \le \frac{91f_k}{r^2} + \eta =\Or( \epsilon'^2),
\end{equation}
with $\Or(r\log\log(r)\log(r/\eta))$ uses of the reflection operators. Now use the fact that $\sum_{k=1}^{N_t} \omega_k=1,\sum_{k=1}^{N_t} \omega_k^2 =\Or(T)$, we can choose $\epsilon'=\epsilon/\sqrt{T}$, so that
\begin{equation}
\abs{ \sum_{k=1}^{N_t}\omega_k \mathbb{E}\wt{f}_k- \sum_{k=1}^{N_t}\omega_k f_k } =\Or(\epsilon^2/T), \qquad \V\left(\sum_{k=1}^{N_t}\omega_k \wt{f}_k \right) =\Or\left( \epsilon'^2\sum_{k=1}^{N_t}\omega_k^2\right)=\Or(\epsilon^2).
\end{equation}
Note that the bias can be neglected as long as $\epsilon/T\ll 1$. Then apply the Chebyshev inequality and then median of means, for any failure probability $0<\delta<1$, we can run the process above for $N_s=\Or(\log \delta^{-1})$ times to obtain an estimator to $\wt{A}$ denoted by $\mc{A}$, so that 
\begin{equation}
\mathbb{P}\left(\abs{\wt{A}-\mc{A}}\ge \frac{\epsilon}{2}\right)\le \delta.
\end{equation}
The total number of uses of reflection operators is thus $\Or((\sqrt{T}\log \delta^{-1})/\epsilon)$.

The total cost of the algorithm is therefore 
\begin{equation}
\Or\Bigl(\sum_{k=1}^{N_t} \|H\|k\Delta t \cdot C_d \cdot N_s \Bigr) = \Or\Bigl( \frac{\|H\|T^{2.5}\log(1/\delta)}{\epsilon} \Bigr)
\end{equation}
queries to the matrix oracle for $H(t)$. 

Regarding the state preparation, we need to prepare a number of quantum states $\Bigl\{ \ket{\psi(t_k)} \Bigr\}_{k=1}^{N_t}$ and each $\ket{\psi(t_k)}$ requires $C_d \cdot N_s$ copies in the circuit of unbiased amplitude estimation. 
In all, the algorithm takes
\begin{equation}
N_t\cdot C_d\cdot N_s = \Or\Bigl(\frac{T^{1.5} \log(1/\delta)}{\epsilon}\Bigr)
\end{equation}
queries to the state preparation oracle for $\ket{\psi_{\mathrm{in}}}$.
\end{proof}

\section{Quantum linear ODE solver for non-unitary dynamics}\label{sec:ode}

We now introduce the third approach for \prb{dense}. We turn to utilize quantum linear ODE solvers~\cite{Ber14,BCOW17,CL19,Kro22,BC22,JLY22a,JLY22b,ALWZ22,ALL23} to produce a Feynman-Kitaev history state of the Hamiltonian system. Such a history state encodes the full time-evolution of the solution, allowing us to perform an amplitude estimation once to evaluate the dense output on the full time interval.

\begin{theorem}[Quantum linear ODE solver with amplitude estimation]\label{thm:lode-ae}
We consider an instance of the quantum dense output problem in \prb{dense}.
There exist a quantum algorithm producing an observable approximating the cost functional $J(u)$ with error $\epsilon\in(0,1)$, succeeding with probability $1-\delta$, with 
\begin{equation}
\Or\Bigl( \frac{\|H\|T^2\log(1/\delta)}{\epsilon} \Bigr).
\end{equation}
queries to the matrix oracle for $H(t)$ and the state preparation oracle for $\ket{\psi_{\mathrm{in}}}$.
\end{theorem}

\begin{proof}
Given an initial condition, we perform a composite  Clenshaw--Curtis quadrature rule (\app{quadrature}) to divide the time interval $[0,T]$ into $N_t = \Or(T\log(1/\epsilon))$ sub-intervals, with $0 = t_0 < t_{1} < \ldots < t_{N_t} = T$, $h_k = t_{k+1}-t_k$. For a time-independent Hamiltonian $H$, we construct $(N_t+ 1)n\times (N_t + 1)n$ linear system
\begin{equation}
L|\Psi\rangle=|B\rangle
\label{eq:linear_system}
\end{equation}
where $L$ is constructed from the block encoding of $H(t)$ (detailed encoding refers to e.g.,~\cite{BC22}), and the quantum states $|\Psi\rangle$ and $|B\rangle$ are
\renewcommand{\arraystretch}{1.4}
\begin{equation}
|\Psi\rangle = [\psi(t_0), \psi(t_1), \cdots, \psi(t_{N_t})]^T, \quad |B\rangle = [\psi_{\mathrm{in}}, 0, \cdots, 0]^T.
\end{equation}
\renewcommand{\arraystretch}{1} 
We can use the quantum linear ODE solvers such as~\cite[Theorem 2]{BC22} to produce the history state 
\begin{equation}
\ket{\Psi} = \frac{1}{\sqrt{N_t+1}} \sum_{k=0}^{N_t}|k\rangle|\psi(t_k)\rangle
\end{equation}
with 
\begin{equation}
\Or\Bigl(\|H\|T\cdot\polylog(1/\epsilon)\Bigr)
\end{equation}
queries to the matrix oracle for $H(t)$ and the state preparation oracle for $\ket{\psi_{\mathrm{in}}}$. Here we use the fact that $g = \frac{\max_{t\in[0,T]}\|\Psi(t)\|}{\|\Psi(T)\|} = 1$ and $\frac{\|\psi_{\mathrm{in}}\|}{\max_{t\in[0,T]}\|\Psi(t)\|} = 1$.

For a time-dependent Hamiltonian $H(t)$, we can employ the quantum time-dependent differential equation solvers based on quantum Dyson series~\cite{BC22}, quantum spectral methods~\cite{CL19,CLO20}, quantum time marching method~\cite{FLT22}, or linear combinations of Hamiltonian simulation/Schrodingerisation~\cite{ALL23,JLY22a,JLY22b}. For instance, we refer the algorithm based on quantum Dyson series~~\cite[Theorem 1]{BC22} to \app{Dyson}. 

Our goal is to estimate the integral  
\begin{equation}
A=\int_0^T \braket{O}_t\ud t = \int_0^T \braket{\psi(t)|O(t)|\psi(t)} \ud t
\end{equation}
Given quadrature nodes $\{t_{1},\ldots,t_{N_t}\}$ and weights $\{\omega_{1},\ldots,\omega_{N_t}\}$ as introduced in \app{quadrature}, we consider
\begin{equation}
\wt{A} = \sum_{k=1}^{N_t} \omega_{k} \braket{O}_{t_{k}}
\end{equation}
as an approximation to $A$ with $\abs{A-\wt{A}}\le \epsilon/2$. 
We define the selection observable $O_{\text{sel}}$ with the block-diagonal form
\begin{equation}
O_{\text{sel}} = \sum_{k=0}^{N_t}|k\rangle\langle k|\otimes \omega_{k} O(t_{k}),
\label{eq:matrixO}
\end{equation}
and we require that $\|O_{\text{sel}}\| \le 1$ (ensured by each $|\omega_k|\le1$ and $\|O(t_k)\| \le 1$), and $O_{\text{sel}}$ can be block-encoded by a unitary $U_{\text{sel}}$ as
\begin{equation}
U_{\text{sel}}=\begin{pmatrix}
O_{\text{sel}} & * \\
* & *
\end{pmatrix}.
\end{equation}

Such a block encoding $U_{\text{sel}}$ can be constructed by associating controlled registers with each block encoding of $O(t_{k})$ as stated in \prb{dense}.

\paragraph{Global Amplitude Estimation:}

We can estimate $\braket{\Psi|O_{\text{sel}}|\Psi}$ using amplitude estimation in \lem{nae}.
To estimate
\begin{equation}
\braket{\Psi|O_{\text{sel}}|\Psi} = \frac{1}{N_t+1}\sum_{k=0}^{N_t}\omega_{k}\braket{O}_{t_{k}} = \frac{1}{N_t+1}\wt A
\end{equation}
within precision $\epsilon/(N_t+1)$ with $N_t+1 = \Or(T\log(1/\epsilon))$,  we need a circuit of depth 
\begin{equation}
C_d = \Or(\frac{N_t}{\epsilon}) = \Or(\frac{T}{\epsilon}).
\end{equation}
Here we can use the above standard estimator even it is biased, since we only use it to estimate a single amplitude and avoid accumulating the biases.
The algorithm for evaluating $A$ takes
\begin{equation}
\Or\Bigl(\|H\|T \cdot C_d \cdot N_s  \Bigr) = \Or\Bigl( \frac{\|H\|T^2\log(1/\delta)}{\epsilon} \Bigr).
\end{equation}
queries to the matrix oracle for $H(t)$ and the state preparation oracle for $\ket{\psi_{\mathrm{in}}}$.
\end{proof}

\section{Quantum linearization algorithm for nonlinear dynamics}\label{sec:linearization}

In recent years, quantum algorithms for nonlinear differential equations have attracted tremendous attention, and several novel quantum linearization approaches have been developed to handle specific nonlinear problems~\cite{LKK20,LPG20,Jos20,DS20,XWG21,AFJ22,Kro22,LLA22,JL22,JLY22,LLL23,LYW23,LEBS23}. A large class of these linearization methods, such as Carleman linearization, are based on the Koopman von Neumann operator theory~\cite{LM98,Mez05,KBBP16,BK22}. Koopman theory forms the foundation of offering an infinite-dimensional linear representation of a general finite-dimensional nonlinear systems with or without including the original states, for which we can perform finite-dimensional truncation and apply quantum linear (differential) equation solvers to efficiently produce the quantum-encoding solutions. 

The generic linearization approach can only approximate well for weakly nonlinear systems. For generic nonlinear systems, the truncation error is hard to control, and there is a known worst case that cannot be efficiently approximated by quantum mechanics~\cite{LKK20}. However, for particular systems, we can offer an \emph{exact} finite-dimensional linear representation that includes the original states as observable functions. This is based on the Koopman Invariant Subspace (KIS)theory~\cite{BBPK16}. 

Koopman Invariant Subspace theory provides an operator-theoretic perspective on dynamical systems. It demonstrates that nonlinear dynamical systems associated with Hamiltonian flows could be analyzed with an infinite-dimensional linear operator, from which it is of great importance to find a finite-rank approximation. In particular, the infinite-dimensional linear representation that includes the original state variables and their polynomials are known as the Carleman embedding/linearization, as a special case of Koopman embedding/linearization. It is quite rare for a dynamical system to admit a finite-dimensional Koopman Invariant subspace that includes the state variables explicitly. Fortunately, our problem model modified below \prb{qdcd} has an exact finite-dimensional linear representation. We in fact offer the first example in quantum mechanics that satisfies the condition, as a contribution to nonlinear dynamical control theory. 

Given the quantum dynamics
\begin{equation}
    \frac{\d }{\d t}\ket{\psi} = -iH(t) \ket{\psi}, 
\end{equation}
we define
\begin{equation}
J(t) = \int_0^{t} \Bigl(\bra{\psi(\tau)}O(\tau)\ket{\psi(\tau)} + \frac{\mu}{2}u^2(\tau)\Bigr) ~\d \tau,
\end{equation}
which satisfies $J(0)=0$, $J(T)=A$ as target, and since $\bra{\psi(t)}O(t)\ket{\psi(t)}$ and $u^2(t)$ are continuous in $t$, we have
\begin{equation}
\frac{\d J(t)}{\d t} = 
\bra{\psi(t)}O(t)\ket{\psi(t)} + \frac{\mu}{2}u^2(t).
\end{equation}

We rewrite the original problem as a system of quantum-driven classical dynamics
\begin{equation}
\frac{\d }{\d t}   
  \begin{bmatrix}
    \ket{\psi} \\
    J \\
  \end{bmatrix} 
=
  \begin{bmatrix}
    -iH \ket{\psi} \\
    \bra{\psi}O\ket{\psi} + \frac{\mu}{2}u^2\\
  \end{bmatrix},
\label{eq:QDCD}
\end{equation}
where $\ket{\psi}\in\C^n$ is a quantum state, $J\in\C$ is a classical cost function, $H \in \C^{n\times n}$ is a Hamiltonian, and $O\in\C^{n\times n}$ is the observable operator. We aim to develop quantum algorithms for the initial value problem of \eq{QDCD} to obtain the final state $J(T)$.

In particular, we assume $O$ is a low-rank observable. Typical instances include the wave-function follower $O(t)=\ket{\phi(t)}\bra{\phi(t)}$~\cite{serban2005optimal} or the projection onto the allowed subspace~\cite{palao2008protecting}. The low-rank observable are also named as the few-body observable~\cite{HK19,HKP20,HCP22}, with potential applications in fidelity estimation~\cite{GHZ89,DKLP02} and entanglement verification~\cite{GT09}. 

We generally assume that $J(t)$ is positive and lower bounded as $J(t) = \Omega(1)$, and $J(t)$ and non-decreasing in terms of $t$ given semi-positive definite $O(t)$. Such class of cost functionals can be time-increasing or time-oscillatory, as illustrated in \app{example}.

We now restate the problem formulation \prb{dense} with additional low-rank assumptions.

\begin{problem}[Quantum low-rank dense output]\label{prb:qdcd}
A time-accumulated observable associated with a time-dependent quantum dynamics is given by
\begin{equation}
\begin{aligned}
    \frac{\d }{\d t}\ket{\psi(t)} &= -iH(t) \ket{\psi(t)}, \quad \ket{\psi(0)}=\ket{\psi_{\mathrm{in}}},\\
    J &= J(T) = \int_0^T \Bigl(\bra{\psi(t)}O(t)\ket{\psi(t)} + \frac{\mu}{2}u^2(t)\Bigr)  ~\d t.
\label{eq:low-rank-dense}
\end{aligned}
\end{equation}
Here $H(t)$ and $O(t)$ represent continuous Hermitian matrices with respect to $t$. We have access to unitaries that block encode $H(t)$ and $O(t)$ for all $t$, and $O(t)$ is (numerically) low-rank in the sense that the Hilbert-Schmidt norm of the observable $O$: $\|O\|_{\mathrm{HS}} \coloneqq \sqrt{\mathrm{Tr}(O^2)}$ is upper bounded by a constant independent of the dimension of the matrix. Additionally, we are provided with a state preparation oracle that prepares the initial state $\ket{\psi_{\mathrm{in}}}$. Our objective is to estimate the value of $J$ with a desired precision of $\epsilon$, within a given time duration $T > 0$.
\end{problem}

We apply the Carleman linearization method to derive the exact linear representation of \eq{QDCD} as
\begin{equation}
\frac{\d }{\d t}   
  \begin{bmatrix}
    J \\
    \ket{\psi}\ket{\psi^\ast} \\
  \end{bmatrix} 
=
  \begin{bmatrix}
    0 & P \\
    0 & (-iH)\otimes I + I\otimes (-iH^*) \\
  \end{bmatrix}  
  \begin{bmatrix}
    J \\
    \ket{\psi}\ket{\psi^\ast} \\
  \end{bmatrix}
+  
  \begin{bmatrix}
    \frac{\mu}{2}u^2 \\
    0 \\
  \end{bmatrix},
\label{eq:problem-kis}
\end{equation}
where $P(t)\in\C^{1\times n^2}$ vectorizes the operator $O(t)$ and satisfies 
\begin{equation}
\bra{\psi}O(t)\ket{\psi} = P(t)\ket{\psi}\ket{\psi^\ast}.
\end{equation}
Note that $\norm{P(t)}=\|O\|_{\mathrm{HS}}=\Or(1)$.

By applying the quantum algorithm in~\cite[Theorem 1]{LKK20} or~\cite[Theorem 4.1]{AFJ22}, we state our main algorithmic result as follows.  

\begin{theorem}[Quantum Carleman linearization with amplitude estimation]\label{thm:node-ae}
We consider an instance of the quantum low-rank dense output problem in \prb{qdcd}.
Assuming $H$ and $O$ (and hence $P$) are time-independent, there exists
a quantum algorithm producing an observable approximating the cost functional $J(u)$ with error $\epsilon\in(0,1)$, succeeding with probability $1-\delta$, with 
\begin{equation}
\Or\Bigl( \frac{\|H\|T\Gamma\log(1/\delta)}{\epsilon} \Bigr)
\end{equation}
queries to the matrix oracle for $H$, $P$ and the state preparation oracle for $\ket{\psi_{\mathrm{in}}}$, where we denote
\begin{equation}
\Gamma \coloneqq \frac{|J(T)|^2+1}{|J(T)|}.
\label{eq:Gamma}
\end{equation}
\end{theorem}

\begin{proof}

We consider the linear system
\begin{equation}
\wh L|\wh \Psi\rangle=|\wh B\rangle.
\label{eq:linear_system_padding}
\end{equation} 
Here we require the history state has the form
\begin{equation}
\begin{aligned}
\ket{\wh \Psi} =& \frac{1}{Q}\Biggl\{
\sum_{k=0}^{N_t}|k\rangle\Bigl[J^k|00\rangle + |\psi^k\rangle|\perp\rangle] + \sum_{k=N_t+1}^{2N_t+2}|k\rangle\Bigl[J(T)|00\rangle + |\psi(T)\rangle|\perp\rangle\Bigr]\Biggr\} \\
=& \frac{1}{Q}\Biggl\{
\sum_{k=0}^{N_t}|k\rangle\Bigl[J^k|0\rangle|0\rangle + \sum_{j,l=1}^n\psi_j^k(\psi_l^{\ast})^k|j\rangle|l\rangle\Bigr] + \sum_{k=N_t+1}^{2N_t+2}|k\rangle\Bigl[J(T)|0\rangle|0\rangle+\sum_{j,l=1}^n\psi_j(T)\psi_l^{\ast}(T)|j\rangle|l\rangle\Bigr]\Biggr\},
\end{aligned}
\label{eq:history_state_padding}
\end{equation}
which includes $N_t+1$ ($N_t = \Or(T\log(1/\epsilon))$) number of the final state $J = J(T)$ to boost the success probability. Here $J^k = J(t^k)$, $|\psi^k\rangle = \sum_{j=1}^n\psi_j^k|j\rangle = \sum_{j=1}^n\psi_j(t^k)|j\rangle$, $|\psi(T)\rangle = \sum_{j=1}^n\psi_j(T)|j\rangle$, $|\!\!\perp\rangle$ is orthogonal to $|00\rangle$, and the normalizing constant is denoted by
\begin{equation}
Q = \sqrt{\sum_{k=0}^{N_t}(|J^k|^2 + 1) + \sum_{k=N_t+1}^{2N_t+2}(|J(T)|^2 + 1)}.
\end{equation}

According to \lem{kis}, we can solve the linear system \eq{linear_system_padding} with 
\begin{equation}
\wt O\Bigl(\|P\|(\|H\|+\|P\|)T\Gamma\cdot\polylog(1/\epsilon)\Bigr)
\end{equation}
queries to the matrix oracle for $H$, $P$ and the state preparation oracle for $\ket{\psi_{\mathrm{in}}}$.

\paragraph{Amplitude Estimation with Padding:}

We consider the amplitude estimation in \lem{nae}, which gives an estimate of
\begin{equation}
a = \braket{\wh \Psi|\wh O|\wh \Psi} = \frac{N_t+1}{Q^2}|J(T)|^2.
\label{eq:a}
\end{equation}
Here 
\begin{equation}
\wh O = \sum_{k=N_t+1}^{2N_t+2}|k\rangle|0\rangle|0\rangle\langle k|\langle 0|\langle 0|
\end{equation}
that indicates the position of $J(T)$.

We now estimate the quantity of $a$ used in the amplitude estimation. On one side, since $J(t)$ is positive,
\begin{equation}
\frac{Q^2}{N_t+1}  \ge \frac{\sum_{k=N_t+1}^{2N_t+2}(|J(T)|^2 + 1)}{N_t+1} = |J(T)|^2+1.
\end{equation}
and on the other side, since $J^k \le J$ for non-decreasing $J(t)$,
\begin{equation}
\frac{Q^2}{N_t+1}  \le \frac{\sum_{k=0}^{N_t}(|J(T)|^2 + 1) + \sum_{k=N_t+1}^{2N_t+2}(|J(T)|^2 + 1)}{N_t+1} = 2|J(T)|^2+2.
\end{equation}

In the amplitude estimation, we aim to produce 
\begin{equation}
\wt a = \frac{N_t+1}{Q^2}|\wt J(T)|^2.
\end{equation}
that approximates $a$, such that $\wt J = \wt J(T)$ is an estimate of $J = J(T)$.
To satisfy $|\wt J - J| \le \epsilon$ with $\epsilon = o(J)$, we require
\begin{equation}
|\wt J^2 - J^2| \le (2J + \wt J - J)|\wt J - J| \le (2J+\epsilon)\epsilon \le 3J\epsilon,
\end{equation}
then it gives
\begin{equation}
|\wt a - a| = | \frac{\wt J^2}{Q^2/(N_t+1)} - \frac{J^2}{Q^2/(N_t+1)} | \le \frac{3J\epsilon}{Q^2/(N_t+1)} \le \frac{3J\epsilon}{J^2 + 1} = \frac{3\epsilon}{\Gamma},
\end{equation}
where $\Gamma$ is denoted as \eq{Gamma}. It suffices to take $t = \Or(\Gamma/\epsilon)$ and the same as the circuit depth
\begin{equation}
C_d = \Or(\frac{\Gamma}{\epsilon}),
\end{equation}
such that $| \wt a - a | \le \frac{3\epsilon}{\Gamma}$, and henceforth $|\wt J - J| \le \epsilon$.

Since the quantum linear ODE solver takes $\Or(\|P\|(\|H\|+\|P\|)T\Gamma)$ to produce the history state $\ket{\wh \Psi}$ in \eq{history_state_padding},  the algorithm for evaluating $A=J(T)$ takes 
\begin{equation}
\Or\Bigl(\|P\|(\|H\|+\|P\|)T\Gamma \cdot C_d \cdot N_s  \Bigr) = \Or\Bigl( \frac{\|P\|(\|H\|+\|P\|)T\Gamma\log(1/\delta)}{\epsilon} \Bigr).
\end{equation}
queries to the matrix oracle for $H$, $P$ and the state preparation oracle for $\ket{\psi_{\mathrm{in}}}$.

Using $\|P\| = \|O\|_{\mathrm{HS}} = \sqrt{\mathrm{Tr}(O^2)} = \Or(1)$, the above complexity can be simplified as 
\begin{equation}
\Or\Bigl( \frac{\|H\|T\Gamma\log(1/\delta)}{\epsilon} \Bigr).
\end{equation}
\end{proof}

This approach achieves the complexity $\Or(T\Gamma/\epsilon)$. Here $\Gamma$ depends on the time-varying behavior of the observables. We refer the examples with different complexity results to \app{example}.

\section{Applications}\label{sec:application}

We introduce several prototype applications of \prb{dense}, including quantum control and spectroscopic computation.

Quantum control plays a pivotal role in the development of quantum technologies such as quantum computing, quantum simulations and quantum sensing.  There are kind classes of control problems: Mayer, Lagrange, and Bolza. A problem of Mayer describes a situation where the cost is determined by the final state and time; a problem of Lagrange describes a situation where the cost accumulates with time; and a problem of Bolza is a combination of problems of Mayer and Lagrange.

We consider a controlled quantum system with a general cost of the Bolza type~\cite{Dal21,BCR10,RWP09,WG07}
\begin{equation}
\begin{aligned}
    \frac{\d }{\d t}\ket{\psi} &= -iH(u(t)) \ket{\psi}, \\
    J(u) &= \int_0^T L\Bigl(\ket{\psi(t)}, u(t), t\Bigr) ~\d t + G\Bigl(\ket{\psi(T)}, T\Bigr).
\end{aligned}
\end{equation}
Here $\ket{\psi}\in\C^n$ is a quantum state, $u\in\C$ is a control function, and $H(u) \in \C^{n\times n}$ is a Hamiltonian determined by the control $u$. We call $L\Bigl(\ket{\psi(t)}, u(t), t\Bigr)$ and $G\Bigl(\ket{\psi(T)}, T\Bigr)$ as running cost and terminal cost.
The problems whose cost functions containing only the terminal or running cost are called Mayer or Lagrange type; and a problem combines both the terminal and running costs is named as Bolza type~\cite{Dal21}.

In our paper, we can express the cost functional as observable functions
\begin{equation}
\begin{aligned}
    L\Bigl(\ket{\psi(t)}, u(t), t\Bigr) &= \bra{\psi(t)}O(t)\ket{\psi(t)} + \frac{\mu}{2}u^2(t) \\
    G\Bigl(\ket{\psi(T)}, T\Bigr) &= \bra{\psi(T)}O(T)\ket{\psi(T)}.
\label{eq:control}
\end{aligned}
\end{equation}
Particularly, we consider certain few-body observables in the running cost, such as the form $O(t) = \ket{\phi(t)}\bra{\phi(t)}$~\cite{HK19,HKP20,HCP22}. When $\ket{\phi(t)} = \ket{\phi}$, it is often used in the state trapping problem, for which we hope $\ket{\psi(t)}$ can approach to the desired state $\ket{\phi}$ as fast as possible and try to stay in that state~\cite{RWP09}. Fidelity estimation with pure target states can also be used in quantum communication (e.g. when $\ket{\phi}$ is a GHZ state~\cite{GHZ89}), fault-tolerant quantum computation (e.g. when $\ket{\phi}$ is a toric code ground st~\cite{GHZ89}), and serve as (bipartite) entanglement witness~\cite{GT09}. In more general, given a desired pure state trajectory $\ket{\phi(t)}$, we can estimate the fidelity $\ket{\phi(t)}\bra{\phi(t)}$ with the change of the evolution time. 

Another formulation of the few-body observable is $O(t) = \sum_{k=1}^r \alpha_k\ket{\phi_k}\bra{\phi_k}$, where $\ket{\phi_k}$ is the $k$-th energy eigenstate of a certain quantum system. We assume $\ket{\psi(t)}$ has a large overlap lying in the eigenspace. Then $|\bra{\psi(t)}O\ket{\psi(t)}|$ can be used to estimate the overlap of $\ket{\psi(t)}$ lying in the low-lying energy eigenspace.

The policy iteration is popular in optimal control or modern predictive control problems. Given a known control $u$, we propagate the controlled quantum dynamics and evaluate the cost functional $J(u)$ in order to update the control $u$. In quantum control, our algorithm can be utilized to estimate $J(u)$ and update $u$. Moreover, it has potential to implement as a subroutine in Variational Quantum Algorithms (VQA)~\cite{CAB21}, Quantum Approximate Optimization Algorithms (QAOA)~\cite{FGG14}, and quantum reinforcement learning~\cite{MUP22}.

Spectroscopic computation is another class of applications that falls into the scope of \prb{dense}. Understanding the spectrum of molecular systems is the first step to understand the effects of molecular aggregates and polymers in computational chemistry. However, it is challenging to quantitatively calculate the spectrum of many-particle dynamics. Several numerical approaches have been developed for computing the spectrum problems based on the Time-Dependent Density Functional Theory (TDDFT) algorithms~\cite{tussupbayev2015comparison} and the Time-Dependent Density Matrix Renormalization Group (TD-DMRG) algorithms~\cite{ren2018time}. 

For instance, we consider the spectra estimation from the quantum dynamics
\begin{equation}
\begin{aligned}
    \frac{\d }{\d t}\ket{\psi} &= -iH \ket{\psi}, \\
    J(\omega) &= \int_0^T e^{i\omega t}\bra{\psi(t)}O(t)\ket{\psi(t)} ~\d t.
\label{eq:spectra}
\end{aligned}
\end{equation}
Here $J(\omega)$ is calculated by taking Fourier transform of the time correlation function, as a type of dense outputs. The factor $e^{i\omega t}$ can be substituted by $\cos(\omega t)$ and $\sin(\omega t)$ instead. 

In zero and finite temperature TD-DMRG, $J(\omega)$ in \eq{spectra} is used to compute the linear absorption and fluorescence spectra of molecular aggregates~\cite{ren2018time}. The spectral analysis of the time correlation function has also been broadly applied to estimate resonance states of molecular systems~\cite{wall1995extraction,mandelshtam1997low,mandelshtam1997spectral}.

\section{Discussion}\label{sec:discussion}

In our work, we manage to estimate the time-accumulated observable associated with the quantum dynamics in \prb{dense} within error tolerance $\epsilon$, and $\epsilon$ is supposed to be independent of the evolution time $T$. In \prb{qdcd} in which we rewrite the original system as a quantum-driven classical dynamics, the cost function $|J(t)|$ can either increases linearly in terms of $t$ when $J(t)>0$ or $J(t)<0$ for all $t$; or $|J(t)|$ can be slowly time-varying, i.e. $J(t) = \Or(\polylog(t))$ for any $t>0$, such as in oscillatory systems. The time-varying behavior of $|J(t)|$ relies on features such as the overlap between the eigenstates of $H(t)$, the observable $O(t)$, and the inhomogeneity (detailed discussions refer to Section 4.5 of~\cite{ALWZ22}). We illustrate simple examples of time-increasing or time-oscillatory observables in \app{example}. In our unified framework, our goal is to upper bound the additive error to $J(t)$ as $\epsilon$ for all cases. The error measurement might change while there is an additional assumption. For instance, it is more desirable to replace $\epsilon$ by $\epsilon T$ for strictly time-increasing $|J(t)|$. It would be of interest to investigate the complexity with respect to different error measurement.

We are concerned with lower bound or fast-forwarding results of \prb{dense}. For general Hamiltonian systems, the no-fast-forwarding theorem gives a $\Omega(T)$ lower bound~\cite{BAC07}, and the Heisenberg limit gives a $\Omega(1/\epsilon)$ lower bound~\cite{AA17}. Henceforth, the upper bound $\Or(T/\epsilon)$ that we achieve in \thm{node-ae} is nearly tight for separately $T$ and $1/\epsilon$. It remains an open problem whether $\Theta(T/\epsilon)$ should be the lower bound for joint $T$ and $1/\epsilon$, or it could be further improved as $\Or(T + T^c/\epsilon)$, $c<1$. Besides, it is appealing to fast-forward particular types of quantum dynamics~\cite{LW18} or non-quantum dynamics~\cite{ALWZ22} and more efficiently produce dense outputs.

We have briefly introduced prototype applications in quantum control and spectroscopic computation. While the output only depends on the single final state, there is a recent developed quantum algorithm for the quantum control problem of the Mayer type~\cite{LW23}. When contemplating applications, we hope future work can investigate quantum algorithms with end-to-end settings. 

.

\section*{Acknowledgments} 

JPL acknowledges the support by the National Science Foundation (PHY-1818914), the NSF Quantum Leap Challenge Institute (QLCI) program (OMA-2016245, OMA-2120757), a Simons Foundation award (No. 825053), and the Simons Quantum Postdoctoral Fellowship. This material is based upon work supported by the U.S. Department of Energy, Office of Science, National Quantum Information Science Research Centers, Quantum Systems Accelerator, and  in  part  by  the  Applied Mathematics Program of the US Department of Energy (DOE) Office of Advanced Scientific Computing Research under contract number DE-AC02-05CH1123 (LL). LL is a Simons Investigator in Mathematics.
We thank Sitan Chen, Soonwon Choi, Di Fang, Aram Harrow, and Seth Lloyd for helpful discussions.

\bibliographystyle{myhamsplain}
\bibliography{KIS}

\appendix

\section{Clenshaw--Curtis quadrature}
\label{app:quadrature}

For a continuous function $f(t)$ defined on $[-1,1]$, the Clenshaw--Curtis quadrature formula approximates the integral $I=\int_{-1}^1 f(t) \ud t=\int_{0}^{\pi} f(\cos\theta) \sin\theta \ud \theta$ by expanding $f(\cos\theta)$ into a polynomial of $\cos\theta$. This amounts to the following quadrature formula
\begin{equation}
I_n=\sum_{k=0}^M \omega_k f(t_k),
\end{equation}
where $t_k$ are chosen to be the Chebyshev points $t_k=\cos \frac{k\pi}{M}, k=0,\ldots,M$. Assuming $M$ is an even number, the weights are
\begin{equation}
\omega_k=\frac{(2-\delta_{k,0}-\delta_{k,M})}{M} \sum_{l=0}^{M/2}(2-\delta_{l,0}-\delta_{l,M/2})\frac{T_{2l}(t_k)}{1-4l^2}, \quad k=0,\ldots,M.
\end{equation} 
Here $T_{l}(x)$ is the $l$-th order Chebyshev polynomial, and $\delta$ is the Kronecker delta. We may write $T_{2l}(t_k)=\cos(2lk\pi/M)$, and the quadrature weights $\{\omega_k\}$ are positive~\cite{Imhof1963}.   
When $M$ is large, the weights $\omega_k$ can be efficiently carried out using fast Fourier transform (FFT) (see e.g.,~\cite{Tre08}). We may use other efficient quadrature schemes, such as the Gauss-Legendre quadrature. However, we find that it is simpler to estimate the $2$-norm of the weights $\{\omega_k\}$ for the Clenshaw--Curtis quadrature needed for the tail bound: Using the fact that $\abs{T_{2l}(t_k)}\le 1$, we have
\begin{equation}
\sum_{k=0}^M \omega_{k}^2\le  \sum_{k=0}^M \frac{16}{M^2}\left(1+\sum_{l=1}^{M/2} \frac{1}{(2l)^2-1}\right)^2
\le  \frac{16}{M^2}\sum_{k=0}^M \left(1+\sum_{l=1}^{\infty} \frac{1}{(2l)^2-1}\right)^2=\frac{36(M+1)}{M^2}=\Or(M^{-1}).
\end{equation}

When $T=\Or(1)$, 
we may map the interval $[0,T]$ to $[-1,1]$ via a linear transformation. For simplicity assume that $\braket{O}_t$ is an analytic function in an open region including the interval $[0,T]$ on which $\abs{\braket{O}_t}\le C$ for some constant $C$.  Then the error of Clenshaw--Curtis quadrature decreases exponentially in $M$ (\cite[Theorem 4.5]{Tre08}). In other words, to achieve additive error $\epsilon$, number of quadrature points is $\Or(\log(1/\epsilon))$. 

For long time integration, mapping the interval $[0,T]$ to $[-1,1]$ introduces additional $T$-dependence in the magnitude of $\abs{\braket{O}_t}$. As a result, we may use a composite Clenshaw--Curtis quadrature, which  
divides the time interval $[0,T]$ into $I$ intervals as $0 = t_0 < t_1 < \ldots < t_{I} = T$ with time step with $\Delta t = t_{i+1} - t_i = \Theta(1)$. Within each segment, we use a Clenshaw--Curtis quadrature with nodes $\{t_{i,0},\ldots,t_{i,M}\}$ and weights $\{\omega_{i,0},\ldots,\omega_{i,M}\}$. Putting all the segments together, we can approximate $A$ as
\begin{equation}
\wt{A}=\sum_{i=1}^{I} \sum_{m=0}^{M} \omega_{i,m} \braket{O}_{t_{i,m}}
\end{equation}
with $\abs{A-\wt{A}}\le \epsilon/2$. With some abuse of notation, we reorder $\wt{A}$ as
\begin{equation}
\wt{A}=\sum_{k=1}^{N_t} \omega_{k} \braket{O}_{t_{k}}
\end{equation}
where $N_t$ is the total number of nodes, with $N_t = (M+1)I = \Or(T\log(1/\epsilon))$. The $2$-norm of the weight satisfies
\begin{equation}
\sum_{k=1}^{N_t}\omega_k^2=\Or(T/M)=\Or(T/\log(1/\epsilon)).
\end{equation}

\section{Quantum time-dependent ODE solver based on Dyson series}
\label{app:Dyson}

We state the complexity of the quantum time-dependent ODE solver developed by Berry and Costa~\cite{BC22} as below.

\begin{lemma}[Theorem 1 of~\cite{BC22}]\label{lem:tdode}
Given an ODE of the form
\begin{equation}
\frac{\d }{\d t}x(t) = A(t) x(t) + b(t), \quad x(0)=x_{\mathrm{in}},
\end{equation}
where $b(t)\in\C^n$ is a vector function of $t$, $A(t)\in\C^{n\times n}$ is a coefficient matrix with non-positive logarithmic norm, and $x(t)\in\C^n$ is the solution vector as a function of $t$. The parameters of the differential equation are provided via $U_A$, $U_b$ and $U_x$ such that 
\begin{equation}
\bra{0}U_A\ket{0} = \frac{1}{\lambda_A}A(t), \quad U_b\ket{0} = \frac{1}{\lambda_b}\ket{b(t)}, \quad U_x\ket{0} = \frac{1}{\lambda_x}\ket{x_{\mathrm{in}}}.
\end{equation}
A quantum algorithm can provide an approximation $\ket{\hat x}$ of the solution $\ket{x(T)}$ satisfying $\|\ket{\hat x}-\ket{x(T)}\| \le \epsilon x_{\max}$ using an average number
\begin{equation}
\Or\Bigl(\mathcal{R} \lambda T \log(1/\epsilon)\Bigr)
\end{equation}
calls to $U_b$, $U_x$, 
\begin{equation}
\Or\Bigl(\mathcal{R} \lambda T \log(1/\epsilon) \log(\lambda T/\epsilon)\Bigr)
\end{equation}
calls to $U_A$, and
\begin{equation}
\Or\Bigl(\mathcal{R} \lambda T \log(1/\epsilon) \log(\lambda T/\epsilon) \bigl[\log(T\mathcal{D}/\lambda\epsilon) + \log(\lambda T/\epsilon)\bigr]\Bigr)
\end{equation}
additional gates, 
where $\lambda = \max\{\lambda_A, \lambda_b/x_{\max}\}$, given constants satisfying 
\begin{equation}
\begin{aligned}
    \mathcal{R} &\ge \frac{x_{\max}}{\|x(T)\|} \frac{\lambda_b/\lambda}{\min_m\|v(m\Delta t, (m-1)\Delta t)\| - \epsilon x_{\max}/(\lambda T)}, \\
    \mathcal{D} &\ge \max_{t\in[0,T]}\|A'(t)\| + \frac{\max_{t\in[0,T]}\|b'(t)\|}{x_{\max}},\\
    x_{\max} &\ge \max_{t\in[0,T]}\|x(t)\|, \\
    b_{\max} &\ge \max_{t\in[0,T]}\|b(t)\|.
\end{aligned} 
\end{equation}
Here 
\begin{equation}
v(t, t_0) = \sum_{k=0}^{\infty} \int_{t_0}^{t} \d t_1 \int_{t_0}^{t_1} \d t_2 \cdots \int_{t_0}^{t_{k-1}} \d t_k A(t_1)A(t_2) \cdots A(t_{k-1})b(t_k),
\end{equation}
and
\begin{equation}
\Delta t = \frac{T}{\lceil \frac{T}{\min(\frac{1}{2\lambda_A},\frac{x_{\max}}{b_{\max}})} \rceil}.
\end{equation}
\end{lemma}

When $A$ is time-independent, the complexity can be simplified as below.

\begin{lemma}[Theorem 2 of~\cite{BC22}]\label{lem:tiode}
Given an ODE of the form
\begin{equation}
\frac{\d }{\d t}x(t) = A x(t) + b, \quad x(0)=x_{\mathrm{in}},
\end{equation}
where $b\in\C^n$ is a vector function of $t$, $A\in\C^{n\times n}$ is a coefficient matrix with non-positive logarithmic norm, and $x(t)\in\C^n$ is the solution vector as a function of $t$. The parameters of the differential equation are provided via $U_A$, $U_b$ and $U_x$ such that 
\begin{equation}
\bra{0}U_A\ket{0} = \frac{1}{\lambda_A}A, \quad U_b\ket{0} = \frac{1}{\lambda_b}\ket{b}, \quad U_x\ket{0} = \frac{1}{\lambda_x}\ket{x_{\mathrm{in}}}.
\end{equation}
A quantum algorithm can provide an approximation $\ket{\hat x}$ of the solution $\ket{x(T)}$ satisfying $\|\ket{\hat x}-\ket{x(T)}\| \le \epsilon x_{\max}$ using an average number
\begin{equation}
\Or\Bigl(\mathcal{R} \lambda T \log(1/\epsilon)\Bigr)
\end{equation}
calls to $U_b$, $U_x$, 
\begin{equation}
\Or\Bigl(\mathcal{R} \lambda T \log(1/\epsilon) \log(\lambda T/\epsilon)\Bigr)
\end{equation}
calls to $U_A$, and
\begin{equation}
\Or\Bigl(\mathcal{R} \lambda T \log(1/\epsilon) \log^2(\lambda T/\epsilon) \Bigr)
\end{equation}
additional gates, 
where $\lambda = \max\{\lambda_A, \lambda_b/x_{\max}\}$, given constants satisfying 
\begin{equation}
\begin{aligned}
    \mathcal{R} &\ge \frac{x_{\max}}{\|x(T)\|} , \\
    x_{\max} &\ge \max_{t\in[0,T]}\|x(t)\|.
\end{aligned} 
\end{equation}
\end{lemma}

Without loss of generality, the above quantum algorithm for time-independent ODEs takes query and gate complexity
\begin{equation}
\Or\Bigl(g\cdot\max\{\|H\|,\lambda_b/x_{\max}\}\cdot T\cdot\polylog(1/\epsilon)\Bigr),
\end{equation}
where
\begin{equation}
g = \frac{x_{\max}}{\|x(T)\|}.
\end{equation}
For time-dependent ODEs, the query and gate complexity includes additional factors from Dyson series as described above.

\section{Quantum ODE solver for quantum-driven classical dynamics}
\label{app:QDCD}

\begin{lemma}[Quantum Carleman linearization algorithm]\label{lem:kis}
For \prb{qdcd}, we consider an instance of \eq{QDCD} with its Carleman linearization as defined in \eq{problem-kis}. We assume $H$ and $O$ (and hence $P$) are time-independent.
There exists a quantum algorithm producing a quantum state proportional to $[\psi(T);J(T)]$ with error at most $\epsilon\le1$, succeeding with probability $\Omega(1)$, with a flag indicating success, with
\begin{equation}
\wt O\Bigl(\|P\|(\|H\|+\|P\|)T\Gamma\cdot\polylog(1/\epsilon)\Bigr)
\end{equation}
queries to the matrix oracle for $H$, $P$ and the state preparation oracle for $\ket{\psi_{\mathrm{in}}}$, where
\begin{equation}
\Gamma = \frac{|J(T)|^2+1}{|J(T)|}
\end{equation}
as denoted in \eq{Gamma}.
\end{lemma}

We consider a quantum linear ODE solver for \eq{problem-kis}
\begin{equation}
\frac{\d }{\d t}   
  \begin{bmatrix}
    J \\
    \ket{\psi}\ket{\psi^\ast} \\
  \end{bmatrix} 
=
  \begin{bmatrix}
    0 & P \\
    0 & Q \\
  \end{bmatrix}  
  \begin{bmatrix}
    J \\
    \ket{\psi}\ket{\psi^\ast} \\
  \end{bmatrix}
+  
  \begin{bmatrix}
    \frac{\mu}{2}u^2 \\
    0 \\
  \end{bmatrix}.
\end{equation}
We denote a skew-Hermitian $Q = (-iH)\otimes I + I\otimes (-iH^*)$ for simplicity, and denote $A = [0, P; 0, Q]$.

Without loss of generality, we are able to shift $H$ so that the eigenvalues of $H$ are lower bounded by $1$. Henceforth, $H$ (and hence $Q$) is invertible, and $\norm{H^{-1}}\le 1$. We also consider a non-decreasing $J(t)$ in terms of $t$ given semi-positive definite $O(t)$, such that $g = \max_{t\in[0,T]}\frac{|J(t)|}{|J(T)|} = 1$.

We observe the diagonalization
\begin{equation}
A=V\Lambda V^{-1} =
  \begin{bmatrix}
    0 & P \\
    0 & Q \\
  \end{bmatrix} 
=
  \begin{bmatrix}
    I & P \\
    0 & Q \\
  \end{bmatrix} 
  \begin{bmatrix}
    0 & 0 \\
    0 & Q \\
  \end{bmatrix} 
  \begin{bmatrix}
    I & -PQ^{-1} \\
    0 & Q^{-1} \\
  \end{bmatrix},
\end{equation}
then the matrix exponential of $A$ has the form
\begin{equation}
e^{At}=Ve^{\Lambda t}V^{-1}=
  \begin{bmatrix}
    I & P \\
    0 & Q \\
  \end{bmatrix} 
  \begin{bmatrix}
    I & 0 \\
    0 & e^{Qt} \\
  \end{bmatrix} 
  \begin{bmatrix}
    I & -PQ^{-1} \\
    0 & Q^{-1} \\
  \end{bmatrix}
=
  \begin{bmatrix}
    I & P(e^{Qt}-I)Q^{-1} \\
    0 & e^{Qt} \\
  \end{bmatrix}.
\end{equation}
Since $\|e^{Qt}\|=1$, we have    
\begin{equation}
\max_{t\in[0,T]}\|e^{At}\| \le 1+2\|P\|\|Q^{-1}\| = O(\|P\|).
\end{equation}
When $P$ and $H$ (and hence $Q$) are time-dependent, it is technically difficult to explicitly upper bound the time-ordering exponential 
\begin{equation}
\max_{t\in[0,T]}\|\mathcal{T}e^{\int_0^tA(s) \d s}\| 
\end{equation}
so we only consider the time-independent case.

For the above ODE with a positive norm of the matrix exponential, we employ the quantum algorithm for the time-independent linear ODEs in~\cite{Kro22}. 

\begin{lemma}[Theorem 7 of~\cite{Kro22}]\label{lem:lnode}
Given an ODE of the form
\begin{equation}
\frac{\d }{\d t}x(t) = A x(t) + b, \quad x(0)=x_{\mathrm{in}},
\end{equation}
and define
\begin{equation}
g \coloneqq \frac{\max_{t\in[0,T]}\|x(t)\|}{\|x(T)\|}, \quad C(A) \coloneqq \sup_{t\in[0,T]} \|\exp(At)\|.
\end{equation}
There exists a quantum algorithm that produces a quantum state $\epsilon$-close to the normalized solution with 
\begin{equation}
\Or\Bigl(gT\|A\|C(A)\cdot \polylog(1/\epsilon)\Bigr),
\end{equation}
queries to the oracles for $A$ and $b$, and gate complexity  is greater by polynomial factors.
\end{lemma}

Since $C(A) = O(\|P\|)$, $\|A\| = \|H\|+\|P\|$, and $g = 1$, there is a quantum algorithm for solving \eq{problem-kis} with 
\begin{equation}
\wt O\Bigl(\|P\|(\|H\|+\|P\|)T\Gamma\cdot\polylog(1/\epsilon)\Bigr)
\end{equation}
queries to the oracles for $A$ and $b$, and gate complexity is greater by polynomial factors.

\section{Simple examples of low-rank observables}
\label{app:example}

To compare the cost of different approaches, we consider two types of low-rank observables: \\
(a) time-increasing low-rank observables, e.g. $O(t)=\ket{\psi(t)}\bra{\psi(t)}$; \\
(b) time-oscillatory low-rank observables, e.g. $O(t)=\cos t \cdot \ket{\psi(t)}\bra{\psi(t)}$.

(a) We compute the cost functional $J$ given the observable $O(t)=\ket{\psi(t)}\bra{\psi(t)}$, 
\begin{equation}
J = J(T) = \int_0^T \bra{\psi(t)}O(t)\ket{\psi(t)}  ~\d t = T.
\end{equation}
This demonstrates a time-increasing function $J(T)$ in terms of $T$.

In more general, we consider the wave-function follower $O(t)=\ket{\phi(t)}\bra{\phi(t)}$~\cite{serban2005optimal}. It is used to force the system to follow a predefined wave-function $\phi(t)$. We assume that $\bra{\phi(t)}\ket{\psi(t)} \ge  \gamma > 0$ for all $t$, i.e. $\ket{\psi(t)}$ has a large overlap with $\ket{\phi(t)}$, then
\begin{equation}
J = J(T) = \int_0^T \bra{\psi(t)}O(t)\ket{\psi(t)}  ~\d t \ge \beta T.
\end{equation}
Here $J(T)$ increases linearly with $T$ as well. 

We compute the parameter $\Gamma$ as defined in \eq{Gamma}
\begin{equation}
\Gamma = \frac{|J(T)|^2+1}{|J(T)|} = \Theta(T).
\end{equation}

(b) We compute the cost functional $J$ given the observable $O(t)=\cos t \cdot \ket{\psi(t)}\bra{\psi(t)}$, 
\begin{equation}
J = J(T) = \int_0^T \bra{\psi(t)}O(t)\ket{\psi(t)}  ~\d t = \sin T.
\end{equation}
It indicates that $J(T)$ oscillates with time and $J(T) = O(1)$ for all $T>0$.

We compute the parameter $\Gamma$ as defined in \eq{Gamma}
\begin{equation}
\Gamma = \frac{|J(T)|^2+1}{|J(T)|} = \Theta(1).
\end{equation}

Overall, we examine the query complexities of our algorithms for the two types of observables, as summarized in \tab{observables}. For (a) time-increasing observables, both the quantum linear ODE solver and the quantum Carleman linearization approaches with amplitude estimation can achieve the best scaling $\Or(T^2/\epsilon)$; For (b) time-oscillatory observables, the quantum Carleman linearization approach with amplitude estimation is superior to other approaches with complexity $\Or(T/\epsilon)$.

\begin{table}[H]
\renewcommand{\arraystretch}{1.5}
    \centering
    \begin{adjustbox}{width=\textwidth}
    \begin{tabular}{c|c|c|c|c}
      \hline\hline
      \textbf{Theorem} & \textbf{Algorithm} & \textbf{Measurement} & \textbf{$O(t)=\ket{\psi(t)}\bra{\psi(t)}$} & \textbf{$O(t)=\cos t \cdot \ket{\psi(t)}\bra{\psi(t)}$} \\
      \hline
      \thm{hs-ht} & Hamiltonian simulation & Hadamard test & $\Or(T^3/\epsilon^2)$ & $\Or(T^3/\epsilon^2)$ \\
      \hline
      \multirow{2}*{\thm{hs-ae}} & \multirow{2}*{Hamiltonian simulation} & Biased amplitude estimation & $\Or(T^3/\epsilon)$ & $\Or(T^3/\epsilon)$ \\
      \cline{3-5}
      &  & Unbiased amplitude estimation & $\Or(T^{2.5}/\epsilon)$ & $\Or(T^{2.5}/\epsilon)$ \\
      \hline
      \thm{lode-ae} & Quantum linear ODE solver & Amplitude estimation & $\Or(T^2/\epsilon)$ & $\Or(T^2/\epsilon)$ \\
      \hline
      \thm{node-ae} & Quantum Carleman linearization & Padding, amplitude estimation & $\Or(T^2/\epsilon)$ & $\Or(T/\epsilon)$ \\
      \hline\hline
    \end{tabular}
    \end{adjustbox}
    \caption{Complexities of quantum algorithms for the time-increasing and  time-oscillatory observables. Here $T$ is the evolution time, and $\epsilon$ is the error tolerance. }
    \label{tab:observables}
\end{table}

\end{document}